\documentclass[11pt,a4paper]{article}

\newcounter{thMM}
\newcounter{leMM}
\newcounter{deFF}
\newcounter{exMP}
\newcounter{prOP}
\newcounter{coRR}
\newenvironment{theorem}[1][Theorem]{\refstepcounter{thMM}\trivlist
   \item[\hskip19pt{\bf #1~\arabic{thMM}.}]\it\hskip3pt}{\endtrivlist}

\newenvironment{definition}[1][Definition]{\refstepcounter{deFF}\trivlist
   \item[\hskip19pt{\bf #1~\arabic{deFF}.}]\it\hskip3pt}{\endtrivlist}

\newenvironment{corollary}[1][Corollary]{\refstepcounter{coRR}\trivlist
   \item[\hskip19pt{\bf #1~\arabic{coRR}.}]\it\hskip3pt}{\endtrivlist}

\newenvironment{proof}[1][Proof]{\begin{trivlist}
\item[\hskip \labelsep {\bfseries #1}]}{\end{trivlist}}

\newcommand{\qed}{\nobreak \ifvmode \relax \else
      \ifdim\lastskip<1.5em \hskip-\lastskip
      \hskip1.5em plus0em minus0.5em \fi \nobreak
      \vrule height0.75em width0.5em depth0.25em\fi}

\newcommand{\Dleft}{[\hspace{-1.5pt}[}
\newcommand{\Dright}{]\hspace{-1.5pt}]}
\newcommand{\SN}[1]{\Dleft #1 \Dright}
\newcommand{\proofend}{\flushright $\qed$}
\newcommand{\Q}{\mathcal{Q}}

\usepackage{amsmath}
\usepackage{amssymb}
\usepackage{dsfont}
\usepackage{hyperref}

\DeclareMathOperator{\Vect}{Vect}

\makeatletter
\renewcommand\@biblabel[1]{#1.}
\makeatother

\begin{document}

\author{Andrew James Bruce \footnote{e-mail: andrewjames.bruce@physics.org}}
\date{\today}
\title{On the relation between exact  QS-manifolds and odd Jacobi manifolds}
\maketitle

\begin{abstract}
In this note  we show that given an exact QS-manifold (a natural generalisation of an exact Poisson manifold) one can associate a family of odd Jacobi structures on the same underlying supermanifold.
\end{abstract}

\begin{small}
\end{small}
\section{Introduction}

Poisson manifolds and their ``superisations" are some of the most fundamental and widespread geometric constructions found through all of mathematical physics. Specifically, Poisson brackets on manifolds and supermanifolds are central to the Hamiltonian description of mechanics, as well as the standard methods of quantisation.  Grassmann odd versions of the  Poisson bracket, known as Schouten brackets in the mathematical literature and antibrackets in the physics literature, have found an important application in gauge theory via the Batalin--Vilkovisky formalism.\\

The algebra of smooth functions on a QS-manifold (in the sense of Voronov \cite{Voronov:2001qf}) comes equipped with a Schouten algebra (a.k.a odd Poisson algebra) together with a homological vector field  that satisfies the Leibniz  rule over the Schouten bracket. Note that the Schouten bracket itself satisfies a Leibniz rule over the supercommutative product of functions. Graded QS-manifolds have been put to good use in describing Lie bialgebroids \cite{Voronov:2001qf}.\\

Odd Jacobi manifolds (see \cite{Bruce2011} for definitions and basic properties) are very similar to QS-manifolds, however the associated odd Jacobi brackets satisfy a \emph{weakened} or \emph{modified} version of the Leibniz rule. One still has a homological vector field satisfying the Leibniz rule over the brackets.\\

Due to the obvious similarities between QS-manifolds and odd Jacobi manifolds is natural to wonder \emph{what direct relations exist  between the two classes of supermanifold}?  \\

 In this short note we show that given an \emph{exact} QS structure on a supermanifold $M$, there exists a pencil of odd Jacobi structures on the same underlying  supermanifold $M$. \emph{Exact} means that the Schouten structure $\widehat{S}\in C^{\infty}(T^{*}M)$ is itself a trivial element in the \emph{Schouten cohomology}. Also, the homological vector field  $\widehat{Q} \in \Vect(M)$ is itself  a trivial element in the  \emph{Q-cohomology} as generated by the Lie derivative along itself $L_{\widehat{Q}}$ (acting on vector fields, though this can be extended to the tensor algebra \cite{Lyakhovich2010}). Such conditions require the existence of a particular even vector field, which will refer to  as the  \emph{homothety vector field} in direct analogy with exact Poisson manifolds \cite{Lichnerowicz1977}. \\

 We take as our inspiration the work of Petalidou \cite{Petalidou2002} who discovered a similar relation between exact Poisson structures and (even) Jacobi structures on the same manifold. \\

\newpage

\noindent \textbf{Preliminaries} \\
All vector spaces and algebras will be $\mathds{Z}_{2}$-graded.   We will generally  omit the prefix \emph{super}. By \emph{manifold} we will mean a \emph{smooth real supermanifold}. We denote the Grassmann parity of an object by \emph{tilde}: $\widetilde{A} \in \mathds{Z}_{2}$. By \emph{even} or \emph{odd} we will be referring explicitly to the Grassmann parity.\\

 A \emph{Poisson} $(\varepsilon = 0)$  or \emph{Schouten} $(\varepsilon = 1)$ \emph{algebra} is understood as a vector space $A$ with a bilinear associative multiplication and a bilinear operation $\{\bullet , \bullet\}: A \otimes A \rightarrow A$ such that:
\begin{list}{}
\item \textbf{Grading} $\widetilde{\{a,b \}_{\varepsilon}} = \widetilde{a} + \widetilde{b} + \varepsilon$
\item \textbf{Skewsymmetry} $\{a,b\}_{\varepsilon} = -(-1)^{(\tilde{a}+ \varepsilon)(\tilde{b}+ \varepsilon)} \{b,a \}_{\varepsilon}$
\item \textbf{Jacobi Identity} $\displaystyle\sum\limits_{\textnormal{cyclic}(a,b,c)} (-1)^{(\tilde{a}+ \varepsilon)(\tilde{c}+ \varepsilon)}\{a,\{b,c\}_{\varepsilon}  \}_{\varepsilon}= 0$
\item \textbf{Leibniz Rule} $\{a,bc \}_{\varepsilon} = \{a,b \}_{\varepsilon}c + (-1)^{(\tilde{a} + \varepsilon)\tilde{b}} b \{a,c \}_{\varepsilon}$
\end{list} \vspace{10pt}
for all homogenous elements $a,b,c \in A$.\\

If the Leibniz rule does not hold identically, but is modified as
\begin{equation}
\{a,bc \}_{\varepsilon} = \{a,b \}_{\varepsilon}c + (-1)^{(\tilde{a} + \varepsilon)\tilde{b}} b \{a,c \}_{\varepsilon} - \{a ,\mathds{1}  \} bc,
\end{equation}

then we have even ($\epsilon = 0)$ or odd ($\epsilon = 1)$ \emph{Jacobi algebras}.\\

A manifold $M$ such that $C^{\infty}(M)$ is a Poisson/Schouten algebra is known as a \emph{Poisson/Schouten manifold}. In particular the cotangent of a manifold comes equipped with a canonical Poisson structure.\\

Let us employ   natural local coordinates $(x^{A}, p_{A})$ on $T^{*}M$, with $\widetilde{x}^{A} = \widetilde{A}$ and $\widetilde{p}_{A} = \widetilde{A}$. Local diffeomorphisms on $M$ induce vector  bundle automorphism on $T^{*}M$ of the form
\begin{equation}
\overline{x}^{A} = \overline{x}^{A}(x), \hspace{30pt} \overline{p_{A}}  = \left(\frac{\partial x^{B}}{\partial \overline{x}^{A}}\right)p_{B}.
\end{equation}

We will in effect use the local description as a \emph{natural vector bundle} to define the cotangent bundle of a supermanifold.  The canonical Poisson bracket on the cotangent is given by

\begin{equation}
\{ F,G \} = (-1)^{\widetilde{A} \widetilde{F} + \widetilde{A}} \frac{\partial F}{\partial p_{A}}\frac{\partial G}{\partial x^{A}} - (-1)^{\widetilde{A}\widetilde{F}}\frac{\partial  F}{\partial x^{A}} \frac{\partial G}{\partial p_{A}}.
\end{equation}\\

Given a vector field $X \in \Vect(M)$ the \emph{symbol} is defined by the replacement $\frac{\partial}{\partial x^{A}} \rightarrow p_{A}$. Thus, $X = X^{A}p_{A} \rightarrow \mathcal{X} = X^{A}p_{A} \in C^{\infty}(T^{*}M)$. Note that the Lie bracket between vector fields becomes the Poisson bracket between the respective symbols. \\

A manifold equipped with an odd vector field $Q \in \Vect(M)$, such that the non-trivial condition $Q^{2}= \frac{1}{2}[Q,Q]=0$ holds, is known as a \emph{Q-manifold} and the vector field $Q$ is known as a \emph{homological vector field} for obvious reasons. \\

\begin{definition}
A \textbf{QS structure} $(\widehat{S}, \widehat{Q})$ on a manifold $M$ consists of
\begin{itemize}
\item an odd function $\widehat{S} \in C^{\infty}(T^{*}M)$, of degree two in fibre coordinates,
\item an odd vector field $\widehat{Q} \in \Vect(M)$,
\end{itemize}
such that the following conditions hold:
\begin{enumerate}
\item the homological condition $\widehat{Q}^{2} = \frac{1}{2} [\widehat{Q},\widehat{Q}]=0$,
\item the invariance condition  $L_{\widehat{Q}}\widehat{S} = 0$,
\item the Schouten  condition $\{\widehat{S},\widehat{S} \}= 0$,
\end{enumerate}
The brackets $\{ \bullet, \bullet \}$ are the canonical Poisson brackets on the cotangent bundle of the manifold.
\end{definition}

The function $\widehat{S} \in C^{\infty}(T^{*}(M))$ is known as the Schouten structure and $\widehat{Q}\in \Vect(M)$ is known as the compatible homological vector field. In short, a QS structure on the manifold provides the algebra of smooth functions $C^{\infty}(M)$ over the manifold with a Schouten bracket,

\begin{equation}
\SN{f,g}_{\widehat{S}} := (-1)^{\widetilde{f}+1} \{ \{ \widehat{S}, f \},g \},
\end{equation}

with $f,g \in C^{\infty}(M)$,  such that the homological vector field  $Q \in \Vect(M)$ satisfies the derivation rule

\begin{equation}
\widehat{Q}\left(\SN{f,g}_{\widehat{S}}\right) = \SN{\widehat{Q}(f), g}_{\widehat{S}} + (-1)^{\widetilde{f}+1} \SN{f, \widehat{Q}(g)}_{\widehat{S}}.
\end{equation}

We will call a manifold with a QS structure a QS-manifold.\\

By suitably relaxing the Schouten condition one arrives at the recently explored odd Jacobi manifolds.

\begin{definition}
An \textbf{odd Jacobi structure} $(S,Q)$ on a manifold $M$  consists of
\begin{itemize}
\item an odd function $S \in C^{\infty}(T^{*}M)$, of degree two in fibre coordinates,
\item an odd vector field $Q \in \Vect(M)$,
\end{itemize}
such that the following conditions hold:
\begin{enumerate}
\item the homological condition $Q^{2} = \frac{1}{2} [Q,Q]=0$,
\item the invariance condition  $L_{Q}S = 0$,
\item the compatibility condition $\{S,S \}= - 2 \Q S $,
\end{enumerate}
 Here $\Q \in C^{\infty}(T^{*}M)$ is the principle symbol or ``Hamiltonian"  of the vector field $Q$.  In local coordinates we have $Q =Q^{A} \frac{\partial}{\partial x^{A}} \rightarrow \mathcal{Q}= Q^{A}p_{A}$. The brackets $\{ \bullet, \bullet \}$ are the canonical Poisson brackets on the cotangent bundle of the manifold.
\end{definition}

The function $S \in C^{\infty}(T^{*}M)$ is known as the almost Schouten structure. A manifold that comes equipped with an odd Jacobi structure will be known as an odd Jacobi manifold. The algebra of smooth functions $C^{\infty}(M)$ over an odd Jacobi manifold is an odd Jacobi algebra where the brackets are provided by the following construction:

\begin{equation}
\SN{f,g}_{J} =  (-1)^{\widetilde{f}+1} \{ \{ S,f \},g    \} - (-1)^{\widetilde{f}+1} \{ \Q, fg \},
\end{equation}
with $f,g \in C^{\infty}(M)$. Again, the homological vector field satisfies a derivation rule over the brackets

\begin{equation}
Q\left(\SN{f,g}_{J}\right) = \SN{Q(f), g}_{J} + (-1)^{\widetilde{f}+1} \SN{f, Q(g)}_{J}.
\end{equation}

A Schouten manifold is simultaneously a QS-manifold and an odd Jacobi manifold with the homological vector field set to identically zero.\\

Note that the  homological and invariance condition for both QS and odd Jacobi manifolds   can be written as
\renewcommand{\labelenumi}{$\arabic{enumi}^{\prime}$.}
\begin{enumerate}
\item $\{\widehat{\Q}, \widehat{\Q} \} =0$ \hspace{15pt} and  \hspace{15pt} $\{\Q,\Q  \}=0$,
\item $\{\widehat{\Q}, \widehat{S}  \}=0$ \hspace{15pt} and  \hspace{15pt} $\{\Q, S  \}=0$.
\end{enumerate}
\renewcommand{\labelenumi}{\arabic{enumi}.}
In effect we \emph{define} the  Lie derivative along a vector field  on $M$ acting on $C^{\infty}(T^{*}M)$ as the Hamiltonian vector field associated with the  respective symbol. \\

\section{Exact QS-manifolds}

\begin{definition}
An \textbf{exact QS-manifold} is the quadruple $(M, \widehat{Q}, \widehat{S}, E)$, where $(M, \widehat{Q}, \widehat{S})$ is a QS-manifold and $E \in \Vect(M)$ is an even vector field, referred as the \textbf{homothety vector field} that  satisfies
\begin{equation}
L_{E} \widehat{S} = - \widehat{S} \hspace{15pt} \textnormal{and} \hspace{15pt} L_{E}\widehat{Q} = -\widehat{Q}.
\end{equation}
\end{definition}

The existence of the homothety vector field on a QS-manifold means that  both $\widehat{\Q}$ and $\widehat{S}$ are exact

\begin{eqnarray}
 \nonumber \{ \mathcal{E}, \widehat{S} \} = -\widehat{S} &\longrightarrow& \widehat{S} = \{\widehat{S}, \mathcal{E}  \},\\
\nonumber \{ \mathcal{E}, \widehat{\Q} \} = -\widehat{\Q} &\longrightarrow& \widehat{\Q} = \{ \widehat{\Q}, \mathcal{E} \},
\end{eqnarray}

with respect to the  operators on $C^{\infty}(T^{*}M)$ they generate. Here $\mathcal{E} \in C^{\infty}(T^{*}M)$ is the symbol of the homothety vector field. In other words, the Schouten structure is itself a trivial element in the \emph{Schouten cohomology}  as generated by $\delta_{\widehat{S}} := \{\widehat{S}, \bullet \}$. The homological structure is similarly  a trivial element in the  cohomology of the operator $L_{\widehat{Q}}$. Poisson cohomology goes back to Lichnerowicz \cite{Lichnerowicz1977}, who also introduced the notion of (even) Jacobi manifolds.  For a discussion of the cohomology of a Q-manifold see \cite{Lyakhovich2010}.\\

Let us proceed to the theorem relating exact QS structures  to odd Jacobi structures on the same underlying manifold.

\begin{theorem}\label{theorem}
Let $(M, \widehat{Q}, \widehat{S}, E)$ be an exact QS-manifold. Then the  pair $(S = \widehat{S} + \mathcal{E}\widehat{\Q},  Q = \widehat{Q})$,
provides an odd Jacobi structure on the manifold $M$.
\end{theorem}

\begin{proof}
The proof requires one to examine the the invariance and compatibility conditions for odd Jacobi structures. The homological condition is given.
\begin{itemize}
\item Writing out the self-Poisson bracket of $S$ one obtains
 \begin{eqnarray}
\nonumber \{S, S \} &=& \{\widehat{S}, \widehat{S} \} + 2 \{ \widehat{S}, \mathcal{E} \}\widehat{\Q} + 2 \mathcal{E}\{ \widehat{S}, \widehat{\Q} \}\\
\nonumber &-& \{ \mathcal{E}, \mathcal{E} \}\widehat{\Q}^{2} - 2 \mathcal{E}\{ \mathcal{E}, \widehat{\Q} \}\widehat{\Q} + \mathcal{E}^{2}\{\widehat{\Q}, \widehat{\Q} \}\\
\nonumber &=& 2 \widehat{\Q}\left(\{ \mathcal{E} , \widehat{S}\}+ \mathcal{E} \{\mathcal{E}, \widehat{\Q} \}\right)\\
\nonumber &=& - 2 \widehat{\Q}\left ( \widehat{S}  + \mathcal{E} \widehat{\Q}\right).
\end{eqnarray}
\item Writing out the Poisson bracket between $\mathcal{Q}$ and $S$ one obtains
\begin{eqnarray}
\nonumber \{\mathcal{Q} , S \} &=& \{ \widehat{\Q} , \widehat{S}\} + \{\widehat{\Q}, \mathcal{E}  \}\widehat{\Q} + \mathcal{E}\{ \widehat{\Q}, \widehat{\Q}\}\\
\nonumber &=&  \widehat{\Q}^{2} =0.
\end{eqnarray}
\end{itemize}
Thus $S$ and $Q$ define an odd Jacobi structure on the manifold $M$.
\proofend
\end{proof}

Via a mild generalisation of the above proof we arrive at the following corollary:

\begin{corollary}
  Associated with any exact QS-structure on $M$ is a pencil of odd Jacobi structures also on $M$ given by  $(S = a \widehat{S} + b \mathcal{E}\widehat{\Q}, Q =  b\widehat{Q})$ where $a,b$ are  even parameters (or just real  numbers).
\end{corollary}

In short, every exact QS-manifold is also an odd Jacobi manifold. Setting $a=b=1$ produces a ``canonical" odd Jacobi structure on $M$. Setting $a=1$ and $b=0$ confirms the notion that a Schouten manifold can be thought of as an odd Jacobi manifold with the trivial homological vector field. Setting $a=0$ and $b=1$ confirms the notion that a Q-manifolds can be thought of as an odd Jacobi manifold with the trivial Schouten structure. In a loose sense, intermediate values of $a$ and $b$ interpolate between the extremes of Schouten manifolds and Q-manifolds understood as examples of odd Jacobi manifolds.

\newpage

\section{Some remarks}

It is well known that given an exact (or homogeneous) Poisson bi-vector  $P \in \mathfrak{X}^{2}(M_{0})$, where   $M_{0}$ is a pure even (classical) manifold,  together with  a $1$-codimensional closed submanifold $N \subset M_{0}$ such that the homothety vector field $E \in \Vect(M_{0})$ is transversal to $N$, then $P$ can be reduced to an even Jacobi structure on $N$. For details see \cite{Dazord1991}, as well as the generalisation of \cite{Grabowski2004}.\\

It is expected that a similar theorem relating (exact) Schouten manifolds and odd Jacobi manifolds differing in dimension by one exists. Details, including the subtleties of  working on supermanifolds requires proper  exploration.\\

However, it is clear that a generalisation of ``Poissonisation" exists; the  ``Schoutenisation" of an odd Jacobi manifold. That is given an odd Jacobi manifold $(M, S, Q)$ one can directly construct a Schouten structure on the manifold $M \otimes \mathds{R}$ viz

\begin{equation}\nonumber
\widehat{S} := e^{-t} \left( S - \mathcal{Q}p \right) \in C^{\infty}(T^{*}(M \otimes \mathds{R})),
\end{equation}

where $(t,p)$ are the natural coordinates on $T^{*}\mathds{R}$. Proving that the above is a Schouten structure follows in a straightforward and direct manner.\\

The question of relating odd Jacobi structures and odd contact structures also requires further exploration. The specific example of constructing an odd Jacobi structure on $\mathcal{M} = \Pi T^{*}M \otimes \mathds{R}^{0|1}$ starting from the canonical odd contact structure on $\mathcal{M}$ is given in \cite{Bruce2011}.

\section*{Acknowledgments}
The author would like to thank Professor Janusz Grabowski for his comments on an earlier draft of this work and for pointing out the paper Grabowski et.al. \cite{Grabowski2004}.


\vfill
\begin{center}
Andrew James Bruce\\
email: \texttt{andrewjames.bruce@physics.org}
\end{center}

\end{document}